\documentclass[12pt]{article}

\usepackage[utf8]{inputenc}
\usepackage[T1]{fontenc}
\usepackage[margin=1in]{geometry}
\usepackage{amsmath,amssymb,amsthm}
\usepackage{graphicx}
\usepackage{booktabs}
\usepackage{algorithm}
\usepackage{algorithmic}
\usepackage{natbib}
\usepackage{xcolor}
\usepackage{hyperref}
\usepackage{float}
\usepackage{setspace}

\newtheorem{assumption}{Assumption}
\newtheorem{proposition}{Proposition}

\newcommand{\mcMldidBiasC}{+0.381}
\newcommand{\mcMldidCateB}{0.236}
\newcommand{\mcMldidCateC}{1.946}
\newcommand{\mcMldidCateD}{1.308}

\newcommand{\mcMldidCovC}{0.33}
\newcommand{\mcNaiveCateC}{2.342}
\newcommand{\mcNaiveCateD}{1.061}
\newcommand{\mcStagBiasC}{-0.005}
\newcommand{\mcStagCateB}{0.480}
\newcommand{\mcStagCateC}{1.887}
\newcommand{\mcStagCateD}{0.837}
\newcommand{\mcStagCovC}{1.00}
\newcommand{\mcTwfeBiasC}{-1.249}

\newcommand{\medATT}{-2.25}
\newcommand{\medATThi}{-1.12}
\newcommand{\medATTlo}{-3.10}

\newcommand{\medCateMean}{-1.90}
\newcommand{\medCateSd}{1.17}

\newcommand{\medCorrInc}{0.46}
\newcommand{\medCorrPov}{-0.55}

\newcommand{\medImpInc}{0.42}
\newcommand{\medImpPov}{0.39}
\newcommand{\medImpRegion}{0.19}

\newcommand{\mpATT}{-0.040}
\newcommand{\mpATThi}{-0.015}
\newcommand{\mpATTlo}{-0.064}

\graphicspath{{figures/}}

\title{A Fixed-Effects Causal Forest for Staggered Adoption,
with an Application to Medicaid Expansion}
\author{Harry Aytug\thanks{Amazon Web Services. Email: haytug@amazon.com.}}
\date{\today}

\begin{document}

\maketitle

\begin{abstract}
Difference-in-differences with staggered adoption identifies group-time average
treatment effects $\mathrm{ATT}(g,t)$ by comparing each cohort to units not yet
treated, which avoids the ``forbidden comparisons'' that bias two-way
fixed-effects estimators when effects are heterogeneous. This paper studies the
covariate-\emph{conditional} version of that object, $\tau_{g,t}(x)$, and
estimates it with a \emph{fixed-effects causal forest}. Within each $(g,t)$
comparison block, the outcome and treatment are residualized on unit and period
fixed effects inside each tree node, and honest causal trees split on
treatment-effect heterogeneity in the covariates. The estimand is not new.
\citet{hatamyar2023mldid} introduced it using a doubly-robust R-learner, and
\citet{imai2023doubly} study it for a single continuous covariate. What we add is
a different way to estimate it. Where those methods remove confounding by
modeling nuisance functions, we remove it by differencing out unit and period
effects within each tree node, following the fixed-effects residualization of
\citet{kattenberg2023causal} and \citet{gavrilova2025did} and carrying it into
the Callaway--Sant'Anna group-time structure. In Monte Carlo experiments the
estimator is the only forest-based method that stays unbiased and correctly
covered for the overall effect under staggered timing with cohort-varying
effects; two-way fixed effects and a pooled causal forest inherit large
forbidden-comparison bias. We apply the method to the \citet{callaway2021difference}
minimum-wage panel as a validation and to the staggered county-level rollout of
the ACA Medicaid expansion, where it recovers an average \medATT{}
percentage-point fall in the uninsured rate and a conditional surface on which
poorer and lower-income counties gained substantially more coverage---
heterogeneity measured along socioeconomic covariates that are not lags of the
outcome.


\end{abstract}

\section{Introduction}

Difference-in-differences (DiD) with staggered treatment timing is among the
most widely used research designs in empirical economics. A recent literature
has shown that the workhorse two-way fixed-effects (TWFE) regression can be
badly biased in this setting when treatment effects are heterogeneous across
adoption cohorts or over time: the TWFE coefficient is a variance-weighted
average of all possible $2\times2$ comparisons, including ``forbidden''
comparisons that use already-treated units as controls, and these can enter with
negative weights \citep{goodman2021difference, dechaisemartin2020two,
sun2021estimating}. \citet{callaway2021difference} (hereafter CS) resolve this by
defining the target as a collection of \emph{group-time average treatment
effects} $\mathrm{ATT}(g,t)$, the effect at period $t$ for the cohort first
treated at $g$. Each $\mathrm{ATT}(g,t)$ is identified from a clean comparison of
that cohort against units not yet treated as of $t$, and the estimates are then
aggregated with transparent weights.

Two limitations motivate this paper. First, $\mathrm{ATT}(g,t)$ is an
\emph{average} over the treated cohort; it does not describe how the effect
varies with observed characteristics. In most applications, though, the question
that matters for policy is which places, firms, or people respond most. Second,
the causal-forest literature that estimates such
conditional average treatment effects (CATEs) nonparametrically
\citep{wager2018estimation, athey2019generalized} was developed under
unconfoundedness and does not, out of the box, respect the parallel-trends
identification or the not-yet-treated comparison structure that staggered DiD
requires.

We study the covariate-conditional group-time effect
\begin{equation}
\tau_{g,t}(x) = \mathbb{E}\!\left[\,Y_t(g) - Y_t(\infty) \mid G = g,\; X = x\,\right],
\qquad t \ge g,
\end{equation}
the CS estimand made heterogeneous in pre-treatment covariates $x$. We propose
to estimate it with a \emph{fixed-effects causal forest} that operates block by
block. For each cohort $g$ and post-period $t \ge g$ we form the same clean
two-group, two-period panel that CS use: cohort $g$ against its not-yet-treated
(or never-treated) controls, observed at the reference period $g-1$ and the
evaluation period $t$. On that block we fit a causal forest that, inside each
tree node, residualizes the outcome and treatment on unit and period fixed
effects before searching for splits that maximize heterogeneity in the treatment
effect. The block-level surfaces are aggregated into event-study and overall
summaries with CS-style weights.

We are explicit about what is and is not new. The estimand $\tau_{g,t}(x)$ was
introduced by \citet{hatamyar2023mldid}, whose MLDID estimator combines the
doubly-robust DiD-CATE of \citet{lu2019estimating} with the CS staggered
framework and not-yet-treated controls. \citet{imai2023doubly} study the same
object and provide doubly-robust uniform confidence bands, but for a single
continuous covariate estimated by local polynomials. We do not claim to have
introduced conditional group-time effects, nor the combination of machine
learning with staggered DiD. What we contribute is a different way to estimate
the object. MLDID removes confounding by modeling nuisance functions, an outcome
regression and generalized propensity scores fit with an R-learner. We instead
difference out unit and period effects within each tree node, using the
fixed-effects residualization of \citet{kattenberg2023causal} and
\citet{gavrilova2025did}. Those authors developed that device for a \emph{common}
adoption date---a single treatment cohort, with event-study dynamics measured
against a fixed base period---rather than for staggered timing; we carry it into
the Callaway--Sant'Anna group-time blocks, so that each cohort is compared only to
units not yet treated and effects may differ across adoption cohorts. The two
approaches are alternatives rather than a re-derivation, and
Section~\ref{sec:mc} benchmarks them directly. Under timing
heterogeneity our estimator has lower bias and valid coverage, while the
doubly-robust engine can be more efficient on the conditional surface when
heterogeneity is smooth in covariates.

One caveat shapes how we present the results. Because each block is a clean
two-group within-fixed-effects comparison, the overall (covariate-averaged)
effect our estimator reports is numerically identical, by construction, to the CS
aggregate. The forest adds value only on the conditional object $\tau_{g,t}(x)$
and its event-study aggregate. We therefore lead with the conditional surface and
treat the scalar as a check against CS rather than as a contribution.

In Monte Carlo experiments across four data-generating processes we benchmark the
estimator against TWFE, the CS estimator, a naively pooled fixed-effects causal
forest, and MLDID. Under homogeneous effects, and under effects that vary only in
covariates, all of the estimators are approximately unbiased on the overall
effect. The differences appear when the effect varies in \emph{timing}, with an
early cohort that gains more and a late cohort that gains less. There TWFE and the
pooled forest inherit large forbidden-comparison bias, roughly $-1.25$ on a true
effect of comparable size and with zero confidence-interval coverage, while our
estimator and CS remain close to unbiased with valid coverage. On the conditional
surface, our estimator has the lowest error of the forest-based methods under
homogeneous, timing-heterogeneous, and nonlinear heterogeneity, and improves on
the pooled forest even under smooth covariate heterogeneity; only MLDID's
penalized-linear learner does better in the one design where the true surface is
linear in the covariates. We then apply the method to the ACA Medicaid expansion.
Across the staggered county-level rollout it recovers an average \medATT{}
percentage-point effect on the county uninsured rate, and a conditional surface on
which poorer, lower-income counties gained the most coverage---heterogeneity
measured along covariates that are not lags of the outcome.

The rest of the paper proceeds as follows. Section~\ref{sec:setup} sets up the
panel and the estimand. Section~\ref{sec:id} states the identifying assumptions
and defines the estimator. Section~\ref{sec:mc} reports the Monte Carlo evidence,
and Section~\ref{sec:app} the applications. Section~\ref{sec:conc} concludes and
discusses the limitations.

\section{Setup and estimand}
\label{sec:setup}

We observe a panel of units $i = 1,\dots,N$ over periods $t = 1,\dots,T$. For
each unit-period we observe an outcome $Y_{it}$, a binary treatment indicator
$D_{it} \in \{0,1\}$, and a vector of time-invariant pre-treatment covariates
$X_i \in \mathbb{R}^p$.

\paragraph{Staggered, absorbing treatment.}
Treatment is an absorbing state: once treated, a unit remains treated. Each
unit's adoption time (cohort) is
$G_i = \min\{t : D_{it} = 1\}$, with $G_i = \infty$ for never-treated units, so
that $D_{it} = \mathbf{1}\{t \ge G_i\}$. Let $\mathcal{G}$ denote the set of
observed finite adoption cohorts.

\paragraph{Potential outcomes.}
Following CS, we index potential outcomes by cohort: $Y_{it}(g)$ is the outcome
unit $i$ would realize at $t$ had it first adopted at $g$, and $Y_{it}(\infty)$
is the never-treated potential outcome. The observed outcome is
\begin{equation}
Y_{it} = Y_{it}(\infty)
+ \sum_{g \in \mathcal{G}} \mathbf{1}\{G_i = g\}\big[Y_{it}(g) - Y_{it}(\infty)\big].
\end{equation}
Effects may depend on the cohort $g$, on calendar time $t$, and on the
covariates $X_i$, which are the object of this paper.

\paragraph{Estimand.}
The target is the group-time conditional average treatment effect on the
treated,
\begin{equation}
\boxed{\;\tau_{g,t}(x) = \mathbb{E}\!\left[\,Y_{it}(g) - Y_{it}(\infty)
\mid G_i = g,\; X_i = x\,\right], \qquad t \ge g.\;}
\end{equation}
This is $\mathrm{ATT}(g,t)$ conditioned on $X_i = x$. Reported objects are
weighted averages of it. The \emph{conditional event study} at exposure
$e = t - g$ is
\begin{equation}
\tau^{\mathrm{es}}_e(x) = \sum_{g \in \mathcal{G}} w_g^{(e)}\, \tau_{g,\,g+e}(x),
\end{equation}
with cohort-share weights $w_g^{(e)}$ restricted to cohorts observed at exposure
$e$; the \emph{overall} conditional effect $\tau^{\mathrm{overall}}(x)$ averages
$\tau_{g,t}(x)$ over post-treatment $(g,t)$ blocks with treated-observation-share
weights, and the scalar overall ATT is its covariate average over the treated.

\section{Identification and estimator}
\label{sec:id}

\subsection{Identifying assumptions}

\begin{assumption}[Conditional parallel trends, not-yet-treated]
\label{ass:pt}
For every $g \in \mathcal{G}$, every $t \ge g$, and reference period $g-1$,
\[
\mathbb{E}[\,Y_{it}(\infty) - Y_{i,g-1}(\infty) \mid G_i = g, X_i = x\,]
= \mathbb{E}[\,Y_{it}(\infty) - Y_{i,g-1}(\infty) \mid G_i \in \mathcal{C}_{g,t}, X_i = x\,],
\]
where the comparison set $\mathcal{C}_{g,t} = \{g' : g' > t\}$ collects units not
yet treated as of $t$ (optionally restricted to never-treated units).
\end{assumption}

\begin{assumption}[No anticipation]
\label{ass:na}
For treated units in cohort $g$, $Y_{it}(g) = Y_{it}(\infty)$ for all $t < g$.
\end{assumption}

\begin{assumption}[Within-block overlap]
\label{ass:overlap}
For each $(g,t)$ and each covariate region used by the forest, both cohort $g$
and the comparison set $\mathcal{C}_{g,t}$ occur with positive probability.
\end{assumption}

Assumption~\ref{ass:pt} is the CS conditional parallel-trends condition imposed
cell by cell. Assumption~\ref{ass:overlap} is more demanding than the overlap
condition of a pooled causal forest: a late cohort with no contemporaneous
not-yet-treated controls in some covariate region is not identified there, and
our estimator reports such blocks as dropped rather than silently extrapolating.
Under Assumptions~\ref{ass:pt}--\ref{ass:overlap}, for $t \ge g$,
\begin{equation}
\tau_{g,t}(x) =
\mathbb{E}[Y_{it} - Y_{i,g-1} \mid G_i = g, X_i = x]
- \mathbb{E}[Y_{it} - Y_{i,g-1} \mid G_i \in \mathcal{C}_{g,t}, X_i = x],
\label{eq:idresult}
\end{equation}
a conditional-on-$x$ difference-in-differences between the cohort and its valid
controls, anchored at the pre-adoption reference $g-1$.

\subsection{The estimator}

The estimator marries the not-yet-treated comparison structure to a
fixed-effects causal forest. For each $(g,t)$ with $t \ge g$:
\begin{enumerate}
\item \textbf{Restrict to a clean block.} Keep cohort-$g$ units and their valid
controls $\mathcal{C}_{g,t}$, at periods $\{g-1, t\}$. Already-treated units of
other cohorts are excluded, which is what removes the forbidden comparison.
\item \textbf{Grow honest causal trees with node-level FE residualization.}
Within each node $\ell$, residualize using only block observations,
$\tilde Y^{(\ell)}_{it} = Y_{it} - \hat\alpha^{(\ell)}_i - \hat\gamma^{(\ell)}_t$
and $\tilde D^{(\ell)}_{it} = D_{it} - \hat\pi^{(\ell)}_i - \hat\rho^{(\ell)}_t$,
by iterative two-way demeaning, then split to maximize heterogeneity in the
local treatment effect and estimate leaf effects honestly on a held-out
subsample.
\item \textbf{Aggregate.} Combine the block surfaces $\hat\tau_{g,t}(x)$ into the
event-study and overall objects with the CS-style weights of
Section~\ref{sec:setup}.
\end{enumerate}

Algorithm~\ref{alg:scffe} states the full procedure.

\begin{algorithm}[htbp]
\caption{Staggered fixed-effects causal forest}
\label{alg:scffe}
\begin{algorithmic}[1]
\REQUIRE panel $\{(Y_{it}, D_{it}, X_i, \text{unit}\ i, \text{period}\ t)\}$;
control rule (not-yet-treated or never-treated)
\STATE infer cohorts $G_i = \min\{t : D_{it}=1\}$ ($G_i=\infty$ if never treated)
\FOR{each cohort $g \in \mathcal{G}$ and period $t \ge g$}
  \STATE reference $\gets g-1$; controls
         $\mathcal{C}_{g,t} \gets \{i : G_i > t\}$ (or $\{i : G_i = \infty\}$)
  \STATE form block: cohort-$g$ units $\cup\ \mathcal{C}_{g,t}$ at periods
         $\{g-1, t\}$
  \IF{cohort or control group empty in block}
    \STATE record $(g,t)$ as dropped (overlap failure); skip
  \ELSE
    \STATE grow honest causal trees on the block; \emph{within each node},
           residualize $Y,D$ on unit and period FE by iterative two-way
           demeaning, split to maximize $\tau$-heterogeneity in $X$, estimate
           leaf effects on the held-out honest subsample
           $\Rightarrow \hat\tau_{g,t}(\cdot)$
  \ENDIF
\ENDFOR
\STATE \textbf{aggregate:} event study
       $\hat\tau^{\mathrm{es}}_e(x)=\sum_g w_g^{(e)}\hat\tau_{g,g+e}(x)$;
       overall ATT via treated-share weights
\STATE \textbf{inference:} unit-level block bootstrap over the whole procedure
       (Section~\ref{sec:inference})
\end{algorithmic}
\end{algorithm}

The engine distinction from MLDID lives in step 2: where MLDID forms a
doubly-robust pseudo-outcome from separately estimated nuisance functions and
runs an R-learner, we difference out unit and period effects \emph{within the
node} and read the local effect off the residualized regression. Neither
requires modeling the propensity of treatment timing; both target
\eqref{eq:idresult}.

It is worth spelling out why a naively pooled forest is biased here.
Applying an ordinary fixed-effects causal forest to the full staggered panel,
pooling all periods and cohorts into one node-level residualization, reintroduces
the forbidden comparison. Within a node the pooled within-FE coefficient
is, by the \citet{goodman2021difference} decomposition, a variance-weighted
average of every $2\times2$ comparison including late-versus-early terms in which
an already-treated cohort serves as the control; when effects grow with exposure
those terms enter with negative weight. Node-level residualization localizes but
does not remove this: each covariate cell inherits its own contaminated
weighting. The block construction removes the offending comparisons.
We do not derive a closed-form expression for this bias; its sign and magnitude
are established quantitatively in Section~\ref{sec:mc}.

\subsection{Consistency}

Within a block, our estimator is an ordinary causal forest applied to
fixed-effects-residualized data, which lets us borrow the consistency theory of
\citet{wager2018estimation} block by block, following the reduction argument of
\citet{gavrilova2025did}. Fix a block $(g,t)$ and let
$\tilde Y_{i} = Y_{it} - Y_{i,g-1}$ be the base-period-differenced outcome for
units in cohort $g$ and its comparison set $\mathcal{C}_{g,t}$, with $W_i$ the
cohort-$g$ membership indicator.

\begin{assumption}[Block regularity]
\label{ass:reg}
Within each block: (i) covariates $X_i$ are continuously distributed on a bounded
support; (ii) $\mathbb{E}[\tilde Y \mid X=x, W=w]$ is Lipschitz in $x$ for
$w\in\{0,1\}$; (iii) conditional second moments of $\tilde Y$ are bounded; and
(iv) observations are independent across units (with dependence within a unit
handled by the unit-level bootstrap of Section~3.4).
\end{assumption}

\begin{proposition}[Block consistency]
\label{prop:consistency}
Under Assumptions~\ref{ass:pt}--\ref{ass:reg}, the honest fixed-effects causal
forest fit on block $(g,t)$ yields a pointwise consistent and asymptotically
normal estimator of $\tau_{g,t}(x)$ as the number of block units grows.
\end{proposition}

\begin{proof}[Sketch]
By the identification result \eqref{eq:idresult}, $\tau_{g,t}(x)$ equals the
conditional difference-in-differences in $\tilde Y$ between cohort $g$ and its
valid controls. Node-level two-way fixed-effects residualization on the two-period
block is, for a two-period panel, algebraically the base-period differencing that
maps $(Y_{it}, Y_{i,g-1})$ to $\tilde Y_i$; the forest thus operates on $\tilde
Y$ with treatment $W$, and under Assumption~\ref{ass:pt} $W$ is unconfounded for
$\tilde Y$ given $X$. Assumption~\ref{ass:reg} is the \citet{wager2018estimation}
regularity condition on the transformed problem, and honesty is imposed by
construction. Theorem~1 and Theorem~3 of \citet{wager2018estimation} then deliver
pointwise consistency and asymptotic normality for the block CATE, exactly as in
Proposition~1 of \citet{gavrilova2025did}. Aggregation across blocks with the
fixed weights of Section~\ref{sec:setup} preserves consistency by the continuous
mapping theorem.
\end{proof}

Two caveats bound this result, both shared with the causal-forest DiD literature.
It is a \emph{block-wise} statement under the block-overlap Assumption~\ref{ass:overlap};
it does not by itself establish valid coverage for the pooled cross-block surface
under finite cohorts, for which we rely on the bootstrap below. And it inherits
the finite-sample attenuation that motivates the recentering and orthogonalization
devices discussed by \citet{athey2019generalized} and \citet{gavrilova2025did};
we address the level through recentering and quantify residual bias by simulation.
The honesty invoked here is an asymptotic device: the consistency and normality
guarantee is for the honest forest as block size grows. In finite two-period
blocks with only a few dozen units, honest sample-splitting can attenuate the
surface more than the bias it removes is worth, which is why the implementation
splits honestly only in blocks large enough to afford it. Table~\ref{tab:honesty}
shows the finite-sample surface comparison is robust to this choice, so nothing in
the empirical results turns on trading the asymptotic guarantee for finite-sample
accuracy in small blocks.

\subsection{Inference}
\label{sec:inference}
We report inference for both objects. For the overall effect we use a unit-level
block bootstrap: resample whole unit series, refit the entire pipeline, and
recompute the aggregate. This captures sampling variability including that of the
cohort weights, and it is our primary interval; in the state-assigned application
of Section~\ref{sec:app} we resample at the treatment-assignment (state) level.
For the conditional surface we report a sampling-based pointwise interval from
the same block bootstrap, refitting the estimator on each resample and
recomputing the fitted surface $\hat\tau(x)$ at the evaluation points, then taking
pointwise percentile bands across draws. These bands reflect sampling uncertainty
in $\tau_{g,t}(x)$. They are wider than the forest half-sample \emph{stability}
bands, which measure only sensitivity to which trees were grown. Because recentering shifts
with each resample, we report the bands around the bootstrap distribution of the
recentered surface. This is the pointwise analogue, for the multivariate forest,
of the uniform bands \citet{imai2023doubly} derive for a single continuous
covariate.

\section{Monte Carlo evidence}
\label{sec:mc}

We simulate staggered panels and benchmark five estimators: TWFE-OLS; the CS
group-time estimator with not-yet-treated controls (average aggregation); a
naively pooled fixed-effects causal forest; MLDID; and the staggered fixed-effects
causal forest of this paper. Four data-generating processes hold fixed a
never-treated control group and staggered cohorts, and vary how the effect
depends on covariates and timing: (a) a homogeneous constant effect; (b) an
effect heterogeneous in covariates but constant across cohorts and exposure; (c)
an effect heterogeneous in \emph{timing}, in which earlier cohorts gain more and
the effect grows with exposure; and (d) an effect that is nonlinear in and
interacts the covariates (a threshold-plus-quadratic surface). Case (c) generates
forbidden-comparison bias; case (d) is the regime a forest should recover better
than a linear learner. Covariates are time-invariant at the unit level. For each
DGP we report, over 200 replications, the bias, root-mean-squared error, and 95\%
confidence-interval coverage of the overall ATT. For the forest and DML
estimators we also report the root-mean-squared error of the fitted conditional
surface against the true per-unit effect (CATE-RMSE). The two block-bootstrap
estimators use the same unit-level bootstrap, so their coverage is comparable.

Within each staggered comparison block the forest uses honest sample-splitting
only when the block is large enough to afford it (at least a few hundred
observations); on the small two-period blocks typical of staggered panels, honest
splitting would estimate leaf effects on a handful of observations and attenuate
the surface, so those blocks are fit without it. In the large blocks of the
Medicaid application all blocks clear the threshold and honest splitting is used
throughout. This adaptive rule is not a tuned knob on which the main result
rests: Table~\ref{tab:honesty} sweeps the honesty regime and shows that under
timing heterogeneity (DGP~c) the staggered forest recovers the surface far better
than the pooled forest at \emph{every} setting, whether honesty is forced on,
forced off, or selected by block size. The clean-block construction, not the
honesty setting, is what drives the advantage there.

\begin{table}[htbp]
\centering\footnotesize
\caption{Monte Carlo performance under staggered adoption (200 replications). Bias, RMSE and 95\% CI coverage are for the overall ATT; CATE-RMSE is the error of the fitted conditional surface $\tau_{g,t}(x)$ over treated observations (forest / DML estimators only). Monte Carlo standard errors (across replications) are in parentheses for the mean quantities (bias and CATE-RMSE). Under timing heterogeneity (DGP~c) the staggered forest is the only \emph{forest}-based estimator with low bias and valid coverage on the overall effect, and it attains the lowest CATE-RMSE on the conditional surface.}
\label{tab:mc}
\setlength{\tabcolsep}{4pt}
\begin{tabular}{llrrrr}
\toprule
DGP & Estimator & Bias (se) & RMSE & Coverage & CATE-RMSE (se) \\
\midrule
(a) homogeneous & TWFE-OLS & -0.002 (0.002) & 0.035 & 0.97 & -- \\
 & Callaway--Sant'Anna & +0.001 (0.003) & 0.043 & 0.96 & -- \\
 & naive CFFE & -0.002 (0.002) & 0.035 & 0.97 & 0.169 (0.001) \\
 & MLDID & +0.179 (0.005) & 0.195 & 0.53 & 0.208 (0.005) \\
 & staggered CFFE & +0.001 (0.003) & 0.043 & 0.85 & 0.141 (0.002) \\
\midrule
(b) heterogeneous in $x$ & TWFE-OLS & -0.002 (0.003) & 0.042 & 1.00 & -- \\
 & Callaway--Sant'Anna & +0.001 (0.003) & 0.043 & 1.00 & -- \\
 & naive CFFE & -0.002 (0.003) & 0.042 & 1.00 & 0.589 (0.004) \\
 & MLDID & +0.154 (0.009) & 0.198 & 0.80 & 0.236 (0.006) \\
 & staggered CFFE & +0.001 (0.003) & 0.043 & 1.00 & 0.480 (0.004) \\
\midrule
(c) heterogeneous in timing & TWFE-OLS & -1.249 (0.007) & 1.253 & 0.00 & -- \\
 & Callaway--Sant'Anna & -0.005 (0.003) & 0.041 & 1.00 & -- \\
 & naive CFFE & -1.249 (0.007) & 1.253 & 0.00 & 2.342 (0.009) \\
 & MLDID & +0.381 (0.008) & 0.399 & 0.33 & 1.946 (0.005) \\
 & staggered CFFE & -0.005 (0.003) & 0.041 & 1.00 & 1.887 (0.006) \\
\midrule
(d) nonlinear-interacted & TWFE-OLS & -0.082 (0.005) & 0.107 & 0.98 & -- \\
 & Callaway--Sant'Anna & -0.005 (0.003) & 0.041 & 1.00 & -- \\
 & naive CFFE & -0.082 (0.005) & 0.107 & 0.98 & 1.061 (0.007) \\
 & MLDID & -0.003 (0.010) & 0.147 & 0.97 & 1.308 (0.009) \\
 & staggered CFFE & -0.005 (0.003) & 0.041 & 1.00 & 0.837 (0.005) \\
\bottomrule
\end{tabular}
\end{table}

Table~\ref{tab:mc} reports the results. Three patterns stand out. First, on the
overall effect, TWFE and the pooled forest carry a bias of \mcTwfeBiasC{} with
\emph{zero} coverage under timing heterogeneity (DGP c), whereas CS and the
staggered forest are unbiased with close-to-nominal coverage; under the other DGPs
all estimators are approximately unbiased.\footnote{The staggered forest's scalar
point estimate equals the CS aggregate by construction, but its interval comes
from the unit-level block bootstrap rather than the CS multiplier bootstrap, so
the two need not have identical coverage in finite samples; in DGP~(a) the
block-bootstrap interval covers at 0.85 against CS's 0.96, reflecting harder
bootstrap calibration on the small two-period blocks. Coverage is at or near
nominal for both in the other designs.} This is the forbidden-comparison result:
the staggered forest is the only forest-based estimator that survives timing
heterogeneity on the aggregate. Second, and this is where the method earns its
keep, on the conditional surface under timing heterogeneity (DGP c) the staggered
forest attains a CATE-RMSE of \mcStagCateC{} against the pooled forest's
\mcNaiveCateC{}---a gap of many Monte Carlo standard errors that the honesty
sweep in Table~\ref{tab:honesty} shows is present at every honesty setting. This
is where the clean-block construction earns its keep: the pooled forest's
node-level pooling lets already-treated units contaminate the control group, and
no honesty setting repairs that. The staggered forest also holds the lowest
CATE-RMSE under homogeneous effects (DGP a) and improves on the pooled forest
under smooth covariate heterogeneity (DGP b). Third, the staggered forest and CS
coincide on the scalar in every DGP, as the construction implies; the forest's
added value is the conditional surface, which CS does not provide.

\subsection{Head-to-head with MLDID}
The fifth estimator, MLDID \citep{hatamyar2023mldid}, shares our estimand but
uses the alternative engine: the \citet{lu2019estimating} doubly-robust score
with machine-learning nuisance functions rather than within-block fixed-effects
residualization. Because it too uses not-yet-treated controls, MLDID avoids the
large forbidden-comparison bias of TWFE and the pooled forest: its DGP~(c) bias
is \mcMldidBiasC{} rather than \mcTwfeBiasC{}, which confirms that both engines solve
the staggered problem. The tradeoff shows up in the comparison between the two.
Under timing heterogeneity (DGP~c) the fixed-effects forest does better on every
axis: lower bias (\mcStagBiasC{} versus \mcMldidBiasC{}),
valid versus poor coverage (\mcStagCovC{} versus \mcMldidCovC{}), and a marginally
better conditional surface (\mcStagCateC{} versus \mcMldidCateC{}). Under \emph{nonlinear} interacted
heterogeneity (DGP~d) the forest-based estimators recover the surface better than
MLDID's penalized-linear learner---CATE-RMSE \mcStagCateD{} for the staggered
forest and \mcNaiveCateD{} for the pooled forest, both below MLDID's
\mcMldidCateD{}---as a tree ensemble should against a linear final stage when the
truth is nonlinear. Here the two forests are close to each other; the separation
is between the forest engines and the linear one, which is the point of the DGP.
The one place the ordering reverses is
pure covariate heterogeneity with a truly linear surface (DGP~b): there MLDID's
penalized-linear final stage fits the linear-in-$x$ effect more efficiently
(CATE-RMSE \mcMldidCateB{} versus our \mcStagCateB{}), while our estimator retains
better coverage. MLDID also under-covers throughout, reflecting the greater
difficulty of calibrated inference for the doubly-robust score relative to the
within-transformation bootstrap. In short, the two engines are alternatives: ours
is preferable when effects vary with adoption timing, when the surface is
nonlinear, and when calibrated coverage matters; the doubly-robust engine is more
efficient only when the conditional surface is smooth and linear and timing is
benign.\footnote{We benchmark against a faithful Python reimplementation of the
Lu--Nie--Wager score underlying MLDID, using random-forest nuisance functions. We
attempted to benchmark against the authors' original \texttt{MLDID} R package
directly; it installs but errors at runtime in current R even on its own bundled
example, apparently a version incompatibility with its CVXR solver dependency, so
we could not use it. Our port omits the minimax reweighting of
\citet{hatamyar2023mldid}, which targets the scalar aggregate; the small positive
scalar bias MLDID shows in the table is attributable to this omission. The
comparison in this paper turns on the \emph{conditional surface} (CATE-RMSE), not
the scalar, so the omission does not affect the head-to-head we draw; and our own
estimator's scalar equals the Callaway--Sant'Anna aggregate by construction,
independently of MLDID.}

\subsection{Sensitivity to the honesty regime}
Because honest sample-splitting trades bias for variance and the two-period
staggered blocks are small, one might worry that the surface comparison hinges on
whether honesty is on. Table~\ref{tab:honesty} settles this by running both
forests under the same honesty regime---forced on, forced off, and adaptive at
several block-size thresholds---and reporting the CATE-RMSE at each. Under timing
heterogeneity (DGP~c) the staggered forest beats the pooled forest by many Monte
Carlo standard errors at \emph{every} regime; the clean-block construction, not
the honesty setting, is what delivers the advantage. Under the nonlinear DGP~(d)
the two forests are close to one another at every regime---consistent with the
head-to-head above, where the meaningful separation is between the forest engines
and MLDID's linear final stage rather than between the two forests.

\begin{table}[htbp]
\centering\footnotesize
\caption{Sensitivity of the conditional-surface error (CATE-RMSE) to the honest-splitting regime, under timing heterogeneity (DGP~c) and nonlinear heterogeneity (DGP~d), 200 replications. Both forests are run under the \emph{same} regime at each row, so the comparison is like-for-like. Monte Carlo standard errors in parentheses. Under DGP~(c) the staggered forest beats the pooled forest by a wide margin at every setting: the clean-block construction, not the honesty knob, drives the advantage.}
\label{tab:honesty}
\begin{tabular}{llrr}
\toprule
DGP & Honesty regime & Staggered CATE-RMSE & Pooled CATE-RMSE \\
\midrule
(c) heterogeneous in timing & honest (forced on) & 1.947 (0.006) & 2.342 (0.009) \\
 & non-honest (forced off) & 1.887 (0.006) & 2.209 (0.009) \\
 & adaptive, threshold 50 & 1.947 (0.006) & 2.342 (0.009) \\
 & adaptive, threshold 100 & 1.947 (0.006) & 2.342 (0.009) \\
 & adaptive, threshold 200 & 1.936 (0.007) & 2.342 (0.009) \\
 & adaptive, threshold 400 & 1.887 (0.006) & 2.342 (0.009) \\
\midrule
(d) nonlinear-interacted & honest (forced on) & 1.179 (0.006) & 1.061 (0.007) \\
 & non-honest (forced off) & 0.837 (0.005) & 0.808 (0.007) \\
 & adaptive, threshold 50 & 1.179 (0.006) & 1.061 (0.007) \\
 & adaptive, threshold 100 & 1.179 (0.006) & 1.061 (0.007) \\
 & adaptive, threshold 200 & 1.076 (0.006) & 1.061 (0.007) \\
 & adaptive, threshold 400 & 0.837 (0.005) & 1.061 (0.007) \\
\bottomrule
\end{tabular}
\end{table}

\section{Applications}
\label{sec:app}

\subsection{Validation: the minimum-wage panel}

As a validation on a familiar dataset, we apply the estimator to the county
minimum-wage / teen-employment panel of \citet{callaway2021difference} (500
counties, 2003--2007; cohorts adopting in 2004, 2006, and 2007, with 309
never-treated counties; a single covariate, log county population). The estimator
recovers the group-time pattern reported by CS: the 2004 cohort's effect deepens
with exposure (about $-0.02$, $-0.08$, $-0.14$ at exposures 0--2), and the
overall ATT is \mpATT{} (95\% CI $[\mpATTlo, \mpATThi]$), a small negative effect
of the minimum wage on log teen employment, matching the sign and magnitude of
the published estimates (Table~\ref{tab:mpdta}). This dataset has too thin a
covariate set to show conditional heterogeneity, which is why we turn to a richer
panel.

\begin{table}[htbp]
\centering
\caption{Validation on the \citet{callaway2021difference} minimum-wage county panel: group-time effects $\widehat{\mathrm{ATT}}(g,t)$ on log teen employment, staggered fixed-effects causal forest. Recovers the published pattern (effects deepen with exposure); overall ATT -0.040 (95\% CI $[-0.064,\ -0.015]$).}
\label{tab:mpdta}
\begin{tabular}{ccr}
\toprule
Cohort $g$ & Period $t$ & $\widehat{\mathrm{ATT}}(g,t)$ \\
\midrule
2004 & 2004 & -0.019 \\
2004 & 2005 & -0.078 \\
2004 & 2006 & -0.136 \\
2004 & 2007 & -0.101 \\
2006 & 2006 & +0.005 \\
2006 & 2007 & -0.041 \\
2007 & 2007 & -0.026 \\
\bottomrule
\end{tabular}
\end{table}

\subsection{The ACA Medicaid expansion}

Our substantive application is the staggered, county-level rollout of the
Affordable Care Act's Medicaid expansion. Because states expanded at different
dates (a large 2014 wave, then smaller cohorts through 2023) and ten states have
not expanded, the design is a clean staggered adoption with a large never-treated
control group. We assemble a county-year panel, 2008--2023, from the Census
Small Area Health Insurance Estimates (SAHIE): the outcome is the county
uninsured rate; treatment is the state expansion, coded absorbing under a
majority-year rule (a state is treated in the first calendar year for which its
expansion was effective at least six months); and the time-invariant
pre-treatment covariates are the county's 2013 poverty rate, its 2013 log median
household income (both from the Census Small Area Income and Poverty Estimates,
SAIPE), and census region. We deliberately condition on socioeconomic covariates
that are \emph{not} lags of the outcome: heterogeneity found along the poverty
rate or income cannot be an artifact of mechanical mean reversion in the uninsured
rate. The 2014 cohort alone contains about 1{,}180 counties and every cohort
exceeds 60, so, unlike a cross-country macro panel, cohorts are large enough for a
forest to learn within-cohort heterogeneity. Table~\ref{tab:covariates} lists the
covariates and Table~\ref{tab:descriptives} reports their 2013 distribution by
cohort.

\begin{table}[htbp]
\centering
\caption{Pre-treatment covariates used in the forest. All are time-invariant, measured before any cohort's expansion.}
\label{tab:covariates}
\begin{tabular}{lll}
\toprule
Covariate & Definition & Source \\
\midrule
Poverty rate & County all-ages poverty rate, 2013, \% & SAIPE (Census) \\
Log median income & Log of county median household income, 2013 & SAIPE (Census) \\
Census region & Four-category region (NE/MW/S/W) & Census \\
\bottomrule
\end{tabular}
\end{table}

\begin{table}[htbp]
\centering\footnotesize
\caption{Descriptive statistics by adoption cohort (county baseline, 2013). Each row is a Medicaid-expansion cohort (year the state's expansion took broad effect under the majority-year rule) or the never-treated group; columns report the county count and 2013 means (standard deviations) of the pre-treatment covariates.}
\label{tab:descriptives}
\begin{tabular}{lrccc}
\toprule
Cohort & Counties & Uninsured \% & Poverty \% & Median income (\$) \\
\midrule
never-treated & 1,070 & 20.3 (5.4) & 19.5 (6.9) & 42,850 \\
2014 & 1,180 & 15.1 (4.8) & 16.1 (6.1) & 48,395 \\
2015 & 169 & 14.3 (2.9) & 13.8 (3.8) & 48,980 \\
2016 & 149 & 22.4 (4.7) & 18.5 (6.9) & 45,235 \\
2019 & 150 & 15.5 (2.9) & 15.2 (6.0) & 50,868 \\
2020 & 73 & 20.1 (5.1) & 14.8 (4.0) & 48,163 \\
2021 & 170 & 18.4 (5.1) & 14.9 (4.7) & 45,495 \\
2022 & 115 & 17.6 (3.1) & 18.7 (5.3) & 40,468 \\
2023 & 66 & 15.0 (3.3) & 16.3 (10.5) & 46,219 \\
\midrule
All & 3,142 & 17.6 (5.5) & 17.2 (6.6) & 46,008 \\
\bottomrule
\end{tabular}
\end{table}

Treatment is assigned at the state level and inherited by counties. The number of
independent treatment switches is therefore the number of expanding states, not
the number of counties; the county cross-section buys resolution for the
conditional surface, not additional identifying variation. Accordingly we cluster
inference at the state level, resampling whole states in the block bootstrap.

\begin{table}[htbp]
\centering
\caption{ACA Medicaid expansion: staggered fixed-effects causal forest, county panel 2008--2023 (3,142 counties, 8 cohorts, 42 group-time blocks, 0 dropped). Effect on the county uninsured rate (percentage points); state-clustered block bootstrap. Covariates are the county poverty rate, log median household income, and census region (none a lag of the outcome).}
\label{tab:medicaid}
\begin{tabular}{lr}
\toprule
Quantity & Value \\
\midrule
Overall ATT (pp) & -2.25 \\
\quad state-clustered 95\% CI & $[-3.10,\ -1.12]$ \\
Event study, exposure 0 / 3 / 6 (pp) & -1.5 / -2.6 / -2.8 \\
Conditional surface $\hat\tau(x)$: mean / sd & -1.90 / 1.17 \\
$\mathrm{corr}(\hat\tau,\ \text{poverty rate})$ & -0.55 \\
$\mathrm{corr}(\hat\tau,\ \text{log median income})$ & +0.46 \\
Feature importance: poverty rate & 0.39 \\
\quad log median income & 0.42 \\
\quad region & 0.19 \\
\bottomrule
\end{tabular}
\end{table}

Table~\ref{tab:medicaid} summarizes the results, and Figure~\ref{fig:es} plots
the event study together with the placebo pre-trends. The overall effect is a
\medATT{} percentage-point reduction in the county uninsured rate (95\% CI
$[\medATTlo, \medATThi]$), within the range of the ACA-expansion
literature, and the event study builds from $-1.5$ at adoption to about $-2.8$ by
the sixth year before attenuating at long horizons where few cohorts and a
thinning control pool
remain. The object of interest, however, is the conditional surface: its spread
is substantial (standard deviation \medCateSd{} on a mean of \medCateMean{}), and
it correlates \medCorrPov{} with the county poverty rate and \medCorrInc{} with
log median household income. The forest's splits load on income (\medImpInc{}) and
poverty (\medImpPov{}), with region secondary (\medImpRegion{}). Poorer counties,
and lower-income counties, gained the most coverage from expansion. This is a form
of heterogeneity that the average $\mathrm{ATT}(g,t)$ cannot express and that the
method recovers directly.

Because these covariates are socioeconomic rather than lags of the uninsured
rate, the gradient is not an artifact of mechanical mean reversion: nothing about
conditioning on the poverty rate forces a county with a larger coverage gain. The
finding is also economically sensible---Medicaid expansion reaches more of the
uninsured where poverty is higher and incomes lower. The method's role is to
recover this gradient nonparametrically and jointly with the dynamics rather than
to impose a functional form.

\subsection{Pre-trends and robustness}
The placebo pre-trend test estimates group-time DiDs at pre-adoption periods,
which should be zero under Assumption~\ref{ass:pt}. It is clean in the periods
that matter (exposures $-2$ to $-4$, at $+0.04$ to $+0.20$ pp) but shows large,
trending placebo effects in the deep pre-period (exposures $\le -9$, i.e.\
2008--2012, up to $-1.98$), which reflect divergent uninsured-rate paths between
expansion and non-expansion states through the Great Recession. This does not
threaten identification. Because the estimator uses the universal base period
$g-1$, every cohort's reference year (2013 for the 2014 wave, later for later
cohorts) lies within a 2011--2023 window, so the recession-era years never enter
an estimation block; they appear only as placebo leads. Table~\ref{tab:robust}
makes this concrete: re-estimating on a trimmed 2011--2023 panel, which drops the
recession-era pre-period entirely, leaves the overall ATT and its interval
unchanged to three decimals, because no estimation block used those years.

We also report a \citet{rambachan2023more} relative-magnitudes sensitivity analysis.
Allowing post-treatment differential trends of up to $\bar M$ times the largest
near-pre-period placebo, the average post-expansion effect remains significantly
negative for every $\bar M$ we consider (through $\bar M = 5$): even if
unobserved trend violations after adoption were five times the largest violation
we observe just before it, the conclusion of a coverage gain would stand. The
effect is therefore robust to economically large departures from parallel trends.

\begin{table}[htbp]
\centering
\caption{Robustness of the Medicaid estimates to dropping the pre-period where parallel trends is violated. The full and trimmed panels give identical overall effects because the universal base period $g-1$ places every cohort's reference year within the trimmed window.}
\label{tab:robust}
\begin{tabular}{lccc}
\toprule
Panel window & Overall ATT (pp) & 95\% CI & Deep pre-trend ($e{=}-12$) \\
\midrule
2008--2023 (main) & -2.25 & $[-3.10,\ -1.12]$ & -0.52 \\
2011--2023 (trimmed) & -2.25 & $[-3.10,\ -1.12]$ & (dropped) \\
\bottomrule
\end{tabular}
\end{table}

\begin{figure}[htbp]
\centering
\includegraphics[width=0.72\textwidth]{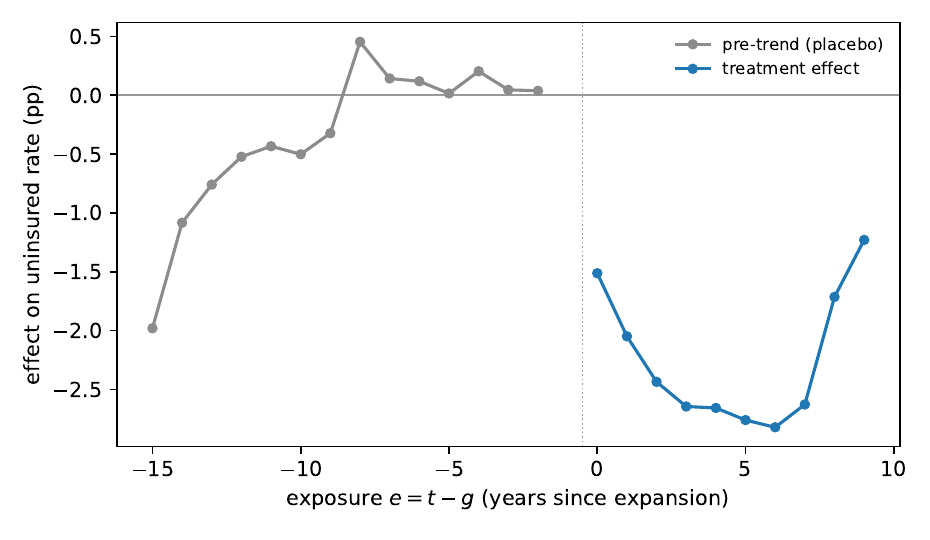}
\caption{Medicaid expansion: event study of the effect on the county uninsured
rate (percentage points) by exposure $e = t - g$. Post-treatment estimates
(right of the reference period) trace the coverage gain; the placebo pre-trend
leads (left) are near zero close to adoption and diverge only in the deep
pre-period, which never enters an estimation block.}
\label{fig:es}
\end{figure}

\section{Conclusion}
\label{sec:conc}

This paper has proposed a fixed-effects causal forest for staggered-adoption
difference-in-differences that delivers covariate-conditional group-time
treatment effects $\tau_{g,t}(x)$. The estimand is shared with recent work
\citep{hatamyar2023mldid, imai2023doubly}. The contribution is the estimation
engine: node-level, within-block two-way fixed-effects residualization inside
honest causal trees, carried into the Callaway--Sant'Anna group-time structure.
We provide Monte Carlo evidence that the estimator stays unbiased and correctly
covered under timing heterogeneity, where two-way fixed effects and a pooled
forest fail, and an application in which it recovers interpretable heterogeneity
in the effect of Medicaid expansion on insurance coverage.

Some limitations remain and point to natural next steps. The consistency result
(Proposition~\ref{prop:consistency}) is block-wise, and we establish the
forbidden-comparison bias of the pooled alternative by simulation rather than in
closed form. Inference for the conditional surface rests on a block bootstrap
(Section~\ref{sec:inference}); sharpening it to the uniform bands
\citet{imai2023doubly} derive for one covariate is a natural extension to the
multivariate forest. The method is a complement to existing staggered-DiD tools:
when the analyst wants to know not just how large the average effect is but for
whom, and is willing to maintain parallel trends rather than unconfoundedness, it
delivers the conditional object within the identification the literature already
uses.

\section*{Data and code availability}
The staggered estimator is implemented in the \texttt{causalfe} Python package.
The functionality this paper introduces is not yet part of a tagged public
release; the replication archive bundles the exact package source used for these
results and installs it locally, so the archive reproduces the paper without
depending on a future release. The archive reproduces every number, table, and
figure from committed inputs: each table and macro is generated by a single script
(\texttt{generate\_paper\_artifacts.py}) from the Monte Carlo output and the
application result files, and a one-command driver (\texttt{reproduce.sh}) either
runs the full pipeline (the Monte Carlo including the MLDID benchmark, the
\citet{callaway2021difference} minimum-wage validation, and the Medicaid
application with its scripted, key-free assembly of the SAHIE county panel from
public Census files) or rebuilds the paper from committed results alone.
All computation is seeded, so results are identical across runs and independent
of the degree of parallelism. The estimator package is installed from the bundled
source (\texttt{pip install -e causalfe}); other dependencies are pinned in
\texttt{requirements.txt}; see \texttt{REPLICATION.md} for the full mapping from
scripts to paper objects.

\bibliographystyle{plainnat}
\bibliography{staggered_cffe}

\begin{thebibliography}{12}
\providecommand{\natexlab}[1]{#1}
\providecommand{\url}[1]{\texttt{#1}}
\expandafter\ifx\csname urlstyle\endcsname\relax
  \providecommand{\doi}[1]{doi: #1}\else
  \providecommand{\doi}{doi: \begingroup \urlstyle{rm}\Url}\fi

\bibitem[Athey et~al.(2019)Athey, Tibshirani, and Wager]{athey2019generalized}
Susan Athey, Julie Tibshirani, and Stefan Wager.
\newblock Generalized random forests.
\newblock \emph{The Annals of Statistics}, 47\penalty0 (2):\penalty0
  1148--1178, 2019.

\bibitem[Callaway and Sant'Anna(2021)]{callaway2021difference}
Brantly Callaway and Pedro H.~C. Sant'Anna.
\newblock Difference-in-differences with multiple time periods.
\newblock \emph{Journal of Econometrics}, 225\penalty0 (2):\penalty0 200--230,
  2021.

\bibitem[de~Chaisemartin and D'Haultf{\oe}uille(2020)]{dechaisemartin2020two}
Cl{\'e}ment de~Chaisemartin and Xavier D'Haultf{\oe}uille.
\newblock Two-way fixed effects estimators with heterogeneous treatment
  effects.
\newblock \emph{American Economic Review}, 110\penalty0 (9):\penalty0
  2964--2996, 2020.

\bibitem[Gavrilova et~al.(2025)Gavrilova, Lang{\o}rgen, and
  Zoutman]{gavrilova2025did}
Evelina Gavrilova, Audun Lang{\o}rgen, and Floris~T. Zoutman.
\newblock Difference-in-difference causal forests, with an application to
  payroll tax incidence in norway.
\newblock \emph{Journal of Applied Econometrics}, 40\penalty0 (7):\penalty0
  727--740, 2025.

\bibitem[Goodman-Bacon(2021)]{goodman2021difference}
Andrew Goodman-Bacon.
\newblock Difference-in-differences with variation in treatment timing.
\newblock \emph{Journal of Econometrics}, 225\penalty0 (2):\penalty0 254--277,
  2021.

\bibitem[Hatamyar et~al.(2023)Hatamyar, Kreif, Rocha, and
  Huber]{hatamyar2023mldid}
Julia Hatamyar, Noemi Kreif, Rudi Rocha, and Martin Huber.
\newblock Machine learning for staggered difference-in-differences and dynamic
  treatment effect heterogeneity.
\newblock \emph{arXiv preprint arXiv:2310.11962}, 2023.

\bibitem[Imai et~al.(2023)Imai, Qin, and Yanagi]{imai2023doubly}
Shunsuke Imai, Lei Qin, and Takahide Yanagi.
\newblock Doubly robust uniform confidence bands for group-time conditional
  average treatment effects in difference-in-differences.
\newblock \emph{arXiv preprint arXiv:2305.02185}, 2023.

\bibitem[Kattenberg et~al.(2023)Kattenberg, Scheer, and
  Thiel]{kattenberg2023causal}
Mark Kattenberg, Bas Scheer, and Jurre Thiel.
\newblock Causal forests with fixed effects for treatment effect heterogeneity
  in difference-in-differences.
\newblock CPB Discussion Paper 452, CPB Netherlands Bureau for Economic Policy
  Analysis, 2023.

\bibitem[Lu et~al.(2019)Lu, Nie, and Wager]{lu2019estimating}
Xiaomeng Lu, Xinkun Nie, and Stefan Wager.
\newblock Estimating individual treatment effects in difference-in-differences
  frameworks.
\newblock \emph{arXiv preprint arXiv:1902.09166}, 2019.

\bibitem[Rambachan and Roth(2023)]{rambachan2023more}
Ashesh Rambachan and Jonathan Roth.
\newblock A more credible approach to parallel trends.
\newblock \emph{Review of Economic Studies}, 90\penalty0 (5):\penalty0
  2555--2591, 2023.

\bibitem[Sun and Abraham(2021)]{sun2021estimating}
Liyang Sun and Sarah Abraham.
\newblock Estimating dynamic treatment effects in event studies with
  heterogeneous treatment effects.
\newblock \emph{Journal of Econometrics}, 225\penalty0 (2):\penalty0 175--199,
  2021.

\bibitem[Wager and Athey(2018)]{wager2018estimation}
Stefan Wager and Susan Athey.
\newblock Estimation and inference of heterogeneous treatment effects using
  random forests.
\newblock \emph{Journal of the American Statistical Association}, 113\penalty0
  (523):\penalty0 1228--1242, 2018.

\end{thebibliography}

\end{document}